\documentclass[11pt]{article} 

\usepackage{fullpage}
\usepackage[margin=1in]{geometry} 

\usepackage{amsmath,amsthm, amssymb}
\usepackage{latexsym}
\usepackage{graphicx}
\usepackage{ifthen}
\usepackage{subcaption}
\usepackage{algorithm}
\usepackage[noend]{algpseudocode}
\usepackage[numbers]{natbib}
\usepackage{mathtools}
\usepackage{dsfont}
\usepackage{nicefrac}
\usepackage{wrapfig}
\usepackage{mfirstuc}
\usepackage{tikz}

\usepackage{hyperref}
\hypersetup{colorlinks=true,linkcolor=blue,citecolor=blue,urlcolor=blue}

\newtheorem*{theorem*}{Theorem}
\newtheorem*{definition*}{Definition}

\newtheorem{lemma}{Lemma}
\newtheorem{definition}{Definition}

\newtheorem{proposition}{Proposition}

\usepackage{thm-restate}

\newcommand{\ron}[1]{{\color{red} Ron: #1}}

\newcommand{\stam}[1]{}

\newcommand{\cupdot}{\mathbin{\mathaccent\cdot\cup}}

\newcommand{\D}{\mathcal{D}}

\newcommand{\R}{\mathbb{R}}
\newcommand{\E}{\mathbb{E}}
\renewcommand{\O}[1]{\mathcal{O}\left( #1\right) }
\newcommand{\Ot}[1]{\tilde{\mathcal{O}}\left( #1\right) }
\newcommand{\Om}[1]{\Omega\left( #1\right) }
\newcommand{\Omt}[1]{\tilde{\Omega}\left( #1\right) }
\newcommand{\Tt}[1]{\tilde{\Theta}\left( #1\right) }
\newcommand{\Th}[1]{{\Theta}\left( #1\right) }

\newcommand{\T}[2]{{\mathcal{U}}_{#1}\left( #2\right) }

\newcommand{\singleedge}[1]{{\mathcal{L}}\left( #1\right) }
\newcommand{\singleedgep}[2]{{\mathcal{L}^{#2}}\left( #1\right) }

\newcommand{\set}[1]{\left\{#1\right\}}

\newcommand\hmm[1]{\ifnum\ifhmode\spacefactor\else2000\fi>1900 \uppercase{#1}\else#1\fi}

\graphicspath{{images/}}

\begin{document}
	
	\title{Finding a Hidden Edge}
	\author{Ron Kupfer\\ Harvard University \\ \href{mailto:ron.kupfer@mail.huji.ac.il}{ron.kupfer@mail.huji.ac.il}
	\and	Noam Nisan\\ The Hebrew University of Jerusalem\\ \href{mailto:noam@cs.huji.ac.il}{noam@cs.huji.ac.il}}
	
	\date{}
	\maketitle	
\begin{abstract}
We consider the problem of finding an edge in a hidden undirected graph $G = (V, E)$ with $n$ vertices, in a model where we only allowed queries that ask whether or not a subset of vertices contains an edge. We study the non-adaptive model and show that while in the deterministic model the optimal algorithm requires $\binom{n}{2}$ queries (i.e., querying for any possible edge separately), in the randomized model $\Tt{n}$ queries are sufficient (and needed) in order to find an edge.\\
In addition, we study the query complexity for specific families of graphs, including Stars, Cliques, and Matchings, for both the randomized and deterministic models.\\
Lastly, for general graphs, we show a trade-off between the query complexity and the number of rounds, $r$, made by an adaptive algorithm. We present two algorithms with $\O{rn^{2/r}}$ and $\Ot{rn^{1/r}}$ sample complexity for the deterministic and randomized models, respectively.
\end{abstract}

\section{Introduction}\label{sec:intro}	
	Consider the following problem: the input is a hidden graph on $n$ vertices and your goal is to identify an edge in the graph.
For this you are allowed to use the following type of query: pick a subset $S$ of the vertices and the answer to the query $S$ is a 
single bit that is true iff there exist two vertices
$u,v \in S$ such that the edge $(u,v)$ is in the hidden graph.  How many such (non-adaptive) queries are needed in order to identify an edge?

These types of questions and similar ones pop up in a surprisingly diverse set of situations. 
Many of them touch on the issue of reductions from search problems (``find an edge'') to decision problems (``is there an edge''). 
The type of computational model that is implied by many of these contexts allows queries that pick a subset of the input bits and ask for their disjunction (or, in a dual model, their conjunction). 
There are several variants of models that differ by the collection of sets that may be used as a query and by whether we allow adaptivity
and randomization.  In Section~\ref{sec:related} we discuss some of these models and their contexts.

\subsection{Warmup}

Before we proceed with our problem let us warm up by looking at the well-studied simplest variant
that ignores the graph structure and views the hidden $m={\binom{n}{2}}$ potential edges of the input graph
simply as a hidden vector $x \in \{0,1\}^m$, and allows accessing the input using queries that pick an
arbitrary $Q \subseteq \{1...m\}$ and ask whether there exists some 
bit $i \in Q$ with $x_i=1$. Our aim is to output some index $i \in \{1...m\}$ with $x_i=1$ (or that no such index $i$ exists, i.e. that $x=00...0$). 

Binary search will certainly solve this problem with $\log_2 (m+1)$ queries which is clearly optimal. Suppose however that we aim for a non-adaptive algorithm, one that makes all the queries before looking at any answer.  What is the non-adaptive complexity? 
It is a simple (and recommended) exercise to prove that one cannot save on querying all $m$ bits, i.e. that the non-adaptive complexity is exactly $m$ \cite{katona2011finding}. A more interesting exercise for the curious reader is to show that we can do better using randomization: 
one may randomly choose $O(\log^2 m)$ non-adaptive queries such that for every $x$, with high probability (over our randomized choices) an 
$i$ with $x_i=1$ is returned \cite{ben1992theory}.\footnote{Hint: start be solving the problem
for the special case where the hidden input contains exactly a single 1 bit.} It turns out that this is optimal and a matching $\Omega(\log^2 m)$ lower bound for randomized non-adaptive algorithms was given in \cite{kawachi2012query}.
	
\subsection{Our Results}
	\stam{
\begin{tabular}{ccc}
	& Deterministic & Randomized \\ 
	\hline 
	General Graphs & $\binom{n}{2}$ & $\O{n\log^2 n}, \Om{\frac{n}{\log^2 n}}$ \\ 
	Star & $\O{n\log n}, \Om{n}$ & $\O{n\log n}, \Om{\frac{n}{\log^2 n}}$ \\ 
	Clique & $\O{n}, \Om{\frac{n}{\log n}}$ & $\Th{\log^2 n}$ (to verify...) \\ 
	Matching & $\O{n\log n}, \Om{n}$ & $\Th{\log^2 n}$ \\ 
\end{tabular}

\begin{table}[h!]
	\centering
\begin{tabular}{ccc}
	& Deterministic & Randomized \\ 
	\hline 
	General Graphs	 & $\binom{n}{2}$ (\cite{katona2011finding})	& $\Tt{n}$ (Thm: \ref{thm:URgeneral}, \ref{thm:LRstar}) \\ 
	Star			 & $\Tt{n}$	(Thm: \ref{thm:UDstar}, \ref{thm:LDstar})		 & $\Tt{n}$  (Thm: \ref{thm:UDstar}, \ref{thm:LRstar}) \\ 
	Clique 			 & $\Tt{n}$ (Thm: \ref{thm:UDclique}, \ref{thm:LDclique})		& $\Tt{1}$ (Thm: \ref{thm:URclique}, \ref{thm:LRclique})\\ 
	Matching 		 & $\Tt{n}$ (Thm: \ref{thm:UDmatching}, \ref{thm:LDmatching})		& $\Tt{1}$ (Thm: \ref{thm:URmatching}, \ref{thm:LRmatching})\\ 
	Graphs using $r$ rounds 		 & $\Th{r\cdot n^{2/r}}$ (Thm: \ref{thm:UDrounds};\cite{gerbner2016rounds})		& $\Ot{r\cdot n^{1/r}}$ (Thm: \ref{thm:URrounds})\\ 
	%$d$-Hypergraphs 		& & $\Omt{n^{d-1}}$	(Thm: ) \\ 
\end{tabular} 
\caption{Summary of our results. References are ordered (upper bound, lower bound).}
%\label{table:results}
\end{table}
}

We now return to our model that does take into account the graph structure and does not allow querying an arbitrary set of edges $Q$ but rather only allow queries that specify a subset of the {\em vertices} and ask whether an edge is contained within this subset.
I.e., in terms of subsets of the edges, we only allow queries of the form $Q_S=\set{(u,v)\mid u,v \in S}$.  
As the trivial upper bound of querying all ${\binom{n}{2}}$ singleton edges is within our query model, and since the deterministic lower bound in the general model of $m={\binom{n}{2}}$ \cite{katona2011finding} still holds, the deterministic simultaneous complexity in our model is
clearly ${\binom{n}{2}}$.  So the basic remaining question is the randomized non-adaptive complexity.

Our main result shows that $\Tt{n}$ queries are needed and sufficient in order to find an edge in the non-adaptive randomized model.
\begin{restatable}{rThm}{General}\label{thm:General}
	There exists a randomized non-adaptive algorithm that finds an edge in a hidden graph, using $\Ot{n}$ queries. The probability of failure of the algorithm on any given input is polynomially small. Moreover, any randomized non-adaptive algorithm that succeeds in finding an edge with at least a constant probability makes at least  $\Omt{n}$ queries.
\end{restatable}

Note that this establishes a quadratic gap between the deterministic and randomized non-adaptive models in our setting.

We continue by analyzing the query complexity for specific families of graphs, both in the randomized and deterministic models.  
Assume that there is a guarantee that the hidden input graph comes from some family of possible graphs.  To what extent does this make
the problem easier?  We study this question for each of three families of graphs: cliques (the graph is known to be a clique on some 
unknown subset of the vertices), stars (a single unknown center vertex connected to some unknown set of vertices), and matchings (the degree of each vertex is at most 1).

\begin{restatable}{rThm}{Special}\label{thm:special}~
	%Finding an edge in a hidden graph requires:
	\begin{itemize}
		\item For Star, Clique, and Matching graphs, $\Tt{n}$ non-adaptive deterministic queries are needed and sufficient in order to find an edge.
		\item For Clique and Matching graphs, there exist non-adaptive randomized algorithms which use only poly-logarithmically many queries.
		\item For Star graphs, any (even randomized) non-adaptive algorithm makes at least $\Tt{n}$ queries.
	\end{itemize}
\end{restatable} 
\noindent

Our final result relaxes the non-adaptivity restriction and considers a model where more than a single round of queries is allowed.
We exhibit a trade-off between the query complexity and the number of rounds.

\begin{restatable}{rThm}{rounds}\label{thm:rounds}
	For finding an edge in a hidden graph using $r$ rounds, $\Th{r\cdot n^{2/r}}$ queries are necessary and sufficient in the deterministic model. There exists a randomized $r$-round algorithm that makes $\Ot{r\cdot n^{1/r}}$ queries.
\end{restatable}
The lower bound for the deterministic case in this theorem was proven by~\cite{gerbner2016rounds}, and together with our randomized
upper bound implies a gap between randomized and deterministic algorithms for any fixed number of rounds.

All three sets of results are summarized in Table~\ref{table:results}.
%\ron{add techniques? derandomization, Yao's principal, isolation lemma}

\begin{table}[h]
	\centering
	\begin{tabular}{ccc}
		& Deterministic & Randomized \\ 
		\hline 
		General Graphs	 & $\binom{n}{2}$ (\cite{katona2011finding})	& $\Tt{n}$  \\ 
		Star			 & $\Tt{n}$			 & $\Tt{n}$   \\ 
		Clique 			 & $\Th{n}$ 		& $\Tt{1}$ \\ 
		Matching 		 & $\Tt{n}$ 	& $\Tt{1}$ \\ 
		General Graphs ($r$ rounds) 		 & $\Th{r\cdot n^{2/r}}$ (lower bound \cite{gerbner2016rounds})		& $\Ot{r\cdot n^{1/r}}$ \\ 
		%$d$-Hypergraphs 		& & $\Omt{n^{d-1}}$	(Thm: ) \\ 
	\end{tabular} 
	\caption{Summary of our results}
	\label{table:results}
\end{table}

\subsection{Paper's Structure}
	In Section~\ref{sec:related} we list some of the contexts in the literature that provide motivation to this model and discuss variants of our model.
In Section~\ref{sec:model} we formally define our problem and add notations.
In Section~\ref{sec:general}, we show how to find an edge in general graphs and prove a matching lower bound.
In Section~\ref{sec:special}, we show algorithms and lower bounds for specific families of graphs, including Star, Cliques, and Matching graphs. 
In Section~\ref{sec:rounds}, we extend our results for general graphs to a model with $r$ adaptive rounds of querying.
We conclude with a discussion and open problems in Section~\ref{sec:discussion}.  
	
\section{Context and Related Work}\label{sec:related}
	Our model belongs to a family of models that allow ``disjunction queries''.  I.e., a query that specifies a subset of the input bits and the answer to the query is a bit that says whether there exists some input bit within the queried set is true in the hidden input, i.e., the 
disjunction of the queried bits.

\subsection{Motivations}

Disjunction-query models (or their dual models that return the conjunction of the input bits) pop up in surprisingly diverse situations. 
Here we list several directions studied before. 

\begin{itemize}
	\item {\bf Reducing search to decision problems.} A disjunction query may be viewed as a decision problem that asks whether some solution {\em exists} within some subset of the space. It is natural to attempt using such queries to actually {\em find} a solution. This can be viewed as the combinatorial core of \cite{valiant1985np} for which a matching lower bound appears in \cite{kawachi2012query}.
	
	\item {\bf Group testing.} Already in World War II, the US public health service found that they can reduce the number of Syphilis tests administered by pooling a set of blood samples, testing for an indicator in the pooled sample, i.e., in the disjunction of the samples. Then, only if the pooled test was positive, continuing with identifying the infected individual \cite{du2000combinatorial}. 
	Similar economies of testing were applied in the recent Covid-19 PCR tests (see e.g. \cite{eberhardt2020multi,gollier2020group}) and a large literature has is devoted to these ideas (see e.g. \cite{du2000combinatorial}).
	
	\item {\bf Learning a hidden graph.} There has been significant work on scenarios where there is some "hidden" graph that needs to be learned, where the available ways to acquire information regarding the graph are exactly these types of disjunction queries (e.g., \cite{grebinski1997optimal,grebinski1998reconstructing,alon2004learning,alon2005learning,angluin2008learning,chang2014learning,abasi2019learning}). Variants' goals are to estimate only basic properties of the hidden graph, such as its number of edges \cite{beame2017edge,chen2020nearly}. 
	
	\item {\bf One-sided decision tree complexity.} In \cite{knop2021log} the following measure of complexity of Boolean decision trees was introduced en route to proving certain lifting theorems: the maximum number of 1-answers returned on a path to a leaf. They show that this measure is equivalent, up to logarithmic factors, to the complexity in a model that allows arbitrary disjunctions but counts the length of paths.\footnote{Technically, they looked at the dual case of 0-answers and conjunctions.}
	
	\item {\bf Demand queries.} In \cite{nisan2021demand} a concrete model of computation that uses economically motivated ``demand queries'' was suggested in order to study algorithms for bipartite matching. It was shown that the demand query model is equivalent, up to logarithmic factors, to a model that allows disjunction queries over subsets of edges that are adjacent to a single left vertex.
	
	\item {\bf Maximizing a valuation function under a cardinality constraint.} We ourselves have stumbled upon this diverse literature when we were trying to study the complexity of maximizing the value of an OXS valuation subject to a cardinality constraint. We describe the question and its relation to disjunction queries in appendix~\ref{app:oxs}.
\end{itemize}

\subsection{Query Models}

Different variants of model differ from each other in the allowed set of queries. To some extent, the graph structure imposed on the hidden input
is a convenient way of specifying several natural classes of allowed queries.
In general in these models, each query specifies a subset of the possible ${\binom{n}{2}}$ edges, where this subset must be
from a pre-specified set of allowed queries. 
There are several natural subsets of allowed queries that have been previously studied:

\begin{itemize}
	
	\item {\bf Singleton:} The weakest model in this family only allows querying single edges, and is thus equivalent to a regular Boolean decision tree (e.g., \cite{messmer1999decision}).
	
	\item {\bf Arbitrary:} The strongest model in this family allows queries the disjunction of any set of edges. This model ignores the graph structure and looks at simply a hidden set of bits (e.g., \cite{assadi2020graph}).
	
	\item {\bf Star:} This model allows queries that look at an arbitrary subset of the edges that are connected to a single vertex, i.e, queries of the form $Q_{v,S} = \set{(u,v)\mid u\in S}$.
	This model was used in \cite{wang2012identifying} {and is equivalent to the demand model studied in \cite{nisan2021demand}}.
	
	\item {\bf Independent Set:} This model allows picking a subset of the vertices and asking whether the set of edges connecting this subset of the vertices contains any edge of the hidden graph. I.e., whether the given subset of vertices is an independent set in the hidden graph. That is, queries of the form $Q_{S} = \set{(u,v)\mid u,v\in S}$.
	This model was used in \cite{grebinski1997optimal,grebinski1998reconstructing,alon2004learning,alon2005learning,angluin2008learning,chang2014learning,abasi2019learning} and is the focus of this paper.
	
	\item {\bf Bi-Set:} This model allows picking two sets of vertices $A$ and $B$ and asks whether the hidden graph contains any edge $(u,v)$ 
	where $u \in A$ and $v \in B$. This model contains the previous two models as special cases and was studied in \cite{beame2017edge,assadi2020graph}.

\end{itemize}

For each of these query models, one may naturally study adaptive vs. non-adaptive algorithms and deterministic vs. randomized ones.

\subsection{Tasks}

This paper is focused on the simplest task of identifying an arbitrary edge in the graph in the Independent Set model. 
The same task was previously studied by \cite{katona2011finding,gerbner2016rounds} in the arbitrary queries model and their lower bounds trivially hold for our case. Katona \cite{katona2011finding} showed that in the non-adaptive case, $\binom{n}{2}$ queries must be made in order of guaranteeing the identification of at least one edge. For $r$ rounds, Gerbner and Vizer \cite{gerbner2016rounds} showed that $\Om{r\cdot n^{2/r}}$ queries are needed.

Other papers have considered other tasks such as completely identifying the hidden graph (e.g., \cite{grebinski1997optimal,grebinski1998reconstructing,alon2004learning,alon2005learning,angluin2008learning,chang2014learning,abasi2019learning}), calculating or approximating the number of edges in it \cite{beame2017edge,chen2020nearly}, or computing some function of the graph  \cite{wang2012identifying, assadi2020graph}.
	
\section{Model and Preliminaries}\label{sec:model}	
	The input for our problem is an undirected graph $G=(V,E)$ with $|V|=n$ and $|E|=m$, for which we are looking for a pair of vertices $u,v$ such that $\set{u,v}\in E$. We denote by $d(v)$ the degree of a vertex $v$, which is the number of its neighbors. 
The access to the input graph is given via an {\em Independent Set Query} (IS), also noted as an {\em Edge-Detecting Query}. % \ronedit{ in other literature}.
\begin{definition}
	An {\em Independent Set (IS) query}, $Q:2^V\rightarrow\set{0,1}$, receives as input a set $A\subseteq V$, and answers whether there exists an edge $(u,v)\in E$ such that both $u$ and $v$ are in $A$. 
	When $A$ contains an edge, we say the query is positive, and negative otherwise.
\end{definition}
A deterministic non-adaptive IS algorithm is a family $\mathcal{F}$ of subsets of $V$, and a mapping from the set of answers $\set{(A,Q(A))\mid A\in \mathcal{F}}$ to a pair of vertices in $V$, or a failing message\footnote{The message could be informative, e.g., ``the graph is empty", ``the graph has at least $k$ edges". These messages are in use when running several tests in parallel, and the additional information is needed for the final algorithm's decision.}.
A randomized IS algorithm is a distribution over deterministic algorithms.
We say that an algorithm is {\em $r$-adaptive} if it consists of $r$ rounds where the choice of queries for the $i$th round may depend on the answers to queries of earlier rounds, but not on answers of the $i$th round onward.
Most of the algorithms and lower bounds in this work are in the non-adaptive model, i.e., $r=1$. Adaptive algorithms are considered only in Section~\ref{sec:rounds}.

In this work we focus on several families of graphs:

\begin{itemize}
	\item \textbf{Singleton graphs}: Graphs with a single edge, i.e., there are $u,v\in V$ such that $E=\set{\set{u,v}}$.
	\item \textbf{Star}: There exists a vertex $v\in V$ and a set $S\subseteq V\setminus\set{v}$, such that $E=\set{\set{s,v}\mid s\in S}$. We refer to $v$ as the {\em center} of the graph.
	\item \textbf{Clique}: There exists a set $S\subseteq V$ such that $E=\set{\set{u,v}\mid u,v\in S}$.
	\item \textbf{Matching}: Each vertex has at most one neighbor.
	\item \textbf{Overlapping-Product}: There exist two, not necessarily disjoint, sets $A,B\subseteq V$, such that $E=\set{\set{a,b}\mid a\in A\ ,b\in B,\ a\neq b}$.
	\footnote{Overlapping-Product graphs is a family containing both Star and Clique graphs and has implications for welfare maximization as we discuss in Appendix~\ref{app:oxs}.}
\end{itemize}

%We also define a generalized {\em Or query} model, in which a query can be made for any subset $S\subseteq \binom{V}{2}$, asking whether $S$ satisfies $|S\cap E| = 0$.
%All our results are with respect to the IS model, although some of the lower bounds we cite are with respect to the more permissive Or model. Note that this is equivalent to the witness finding problem defined in the related work Section.

We denote by $\binom{V}{k}$ the family of subsets of size exactly $k$ contained in $V$, and by $\mathcal{P}(V)$ the family of all subsets of $V$.
We define a distribution $\T{p}{A}$ over subsets of $A$, by taking each item $a\in A$ with independent probability $p$. That is, for $U\sim\T{p}{A}$ and any set $S\subseteq A$, we have that $\Pr\left(U=S\right) = p^{|S|}(1-p)^{|A\setminus S|}$.
We use $\log n$ to denote the logarithm function with base $2$, and $\ln n$ for the natural logarithm.
The following standard asymptotic notions are in use: ${\mathcal{O}},{\Omega}, {\Theta}$, and their equivalents which suppress of polylogarithmic factors of $n$, $\tilde{\mathcal{O}},\tilde{\Omega}, \tilde{\Theta}$.

\section{Finding an Edge in a Hidden Graph}\label{sec:general}
    \subsection{Finding an edge for singleton graphs}
We start by showing that in the case that the hidden graph is known to contain a single edge, the problem is relatively simple and can be solved deterministically using a logarithmic number of queries. 

\begin{restatable}{rPro}{single}\label{pro:single_edge}
	There exists a deterministic IS algorithm that for any graph $G = (V,E)$ with $m$ edges, returns the unique edge $e\in E$, if $m=1$.
	In case that $m=0$ the algorithm returns the message {\em''none"}, and in case that $m>1$ returns the message {\em ''more than one"}, and makes $\O{\log n}$ non-adaptive deterministic queries.
\end{restatable}
\begin{proof}	
	We start with the assumption that 
	$m=1$ and describe a randomized algorithm for the problem of finding this unique edge.
	The algorithm first samples $t=24\ln n$ random sets of vertices according to $\T{\frac{1}{2}}{V}$.
	For every triplet $a,b,c\in V$, the probability that a ${a,b}\in S$ and ${c}\notin S$ is $\frac{1}{8}$, and the probability that this would not hold for any samples is  
	$$\left( 1-\frac{1}{8}\right)^t < n^{-3}.$$ 
	Assuming the graph's single edge is $(a,b)$, for every sample $S$ it holds that $Q(S)$ is positive if and only if $\set{a,b}\subseteq S$.
	Whenever we have a positive $S$ such that $c\notin S$, we are guaranteed that $c$ is not a part of the single edge. Taking a union over all possible $a,b,c\in V$, gives that with positive probability, using $t$ samples, for {\em any} single edge graph $\set{a,b}$, we can rule out {\em all} other vertices as being part of the edge.
	Hence, there exists some selection of the $t$ queries that is correct for all possible single-edge inputs, and we can find the edge deterministically using $t$ queries.
	The full algorithm needs also to distinguish between the cases where $m=0$ and $m>1$. This is in fact a simple task.
	The existent of at least one edge can be verified by the single query $Q(V)$. If there is more than one edge in the graph, the elimination process would always keep at least two pairs of vertices as candidates for an being an edge and will fail returning a single edge, %ruling out $n-2$ vertices in the previous part of the algorithm is impossible
	no matter how many queries the algorithm makes, and this case is also easily identified after the $t$ queries.
	%at least $1-n^{-1}$, for any $c\in V\setminus\set{a,b}$, there is an index $i$ such that $Q(S_i)=1$ and $c\notin S_i$ and we can deduce that this vertex is part of the single edge in the graph.
	%
	%Since, this there are only $\binom{n}{2}$ possible Singleton graphs, there is a positive probability for a random sample to succeed for \textbf{all} of the inputs. 
\end{proof} 
We later use the above algorithm as a sub-procedure in many of the algorithms in this work, and use the following notion:
\begin{definition}
	Given a subset $S\subseteq V$, we denote by $\singleedge{S}$ the set of queries asked in order to solve the problem described in Proposition~\ref{pro:single_edge}, for the graph $G=\left(S,E\cap\binom{S}{2}\right)$.\\
\end{definition}
For Singleton graphs, i.e., when it is known that $m=1$, we give an explicit construction using $\lceil4\log n\rceil$ queries in Appendix~\ref{sec:exlicit}. In addition, for cases when the single edge has a known endpoint $v$, we describe an algorithm that uses $\lceil2\log n\rceil$ queries for the same problem.
When is it possible to use the later, more efficient algorithm, we use the notation $\singleedgep{S}{v}$. Both algorithms have the same asymptotic behavior like the one in Proposition~\ref{pro:single_edge}, and we may use it without harming the asymptotic correctness of any of our claims.

\subsection{Upper Bound for General graphs}
For general graphs, the trivial deterministic algorithm who queries all pairs and returns any pair with positive answer, is in fact optimal as proven in~\cite{katona2011finding}:

\begin{restatable}{rLem}{LDgeneral}\label{thm:LDgeneral}
	\textbf{\cite{katona2011finding}:} Any deterministic algorithm for edge finding in a hidden graph, requires $\binom{n}{2}$ queries.
\end{restatable}
We reprove the claim in Appendix~\ref{sec:LDG}.
On the other hand, we now show that in the randomized model, there exists an algorithm that makes only $\Tt{n}$ queries, and fails with a polynomially small probability.

Our algorithm handles three different cases in parallel. First, if the number of edges in the graph is large, sampling random pairs of vertices finds an edge with high probability. Second, assuming at least one of the vertices $v$ has a high degree, then $v$ and one of its neighbors can be found using an efficient scheme described below. Lastly, if the former two cases do not apply, i.e., there are few edges and they are scattered, we can find an edge by querying few large sets. Whenever receiving a positive answer, there is a good probability there is a single edge in this set. The edge itself can be then identified using the algorithm from Proposition~\ref{pro:single_edge}.

Let $D = \set{2^i\cdot\sqrt{n} \mid i=1,...,\left\lfloor\frac{1}{2}\log n\right\rfloor }$ be a set of estimations for the maximal degree of a vertex in the graph and $c>0$ some constant later to be chosen.
We query the following three families in parallel, each with quasi-linear many queries.
\begin{itemize}
	\item $\mathcal{F}_1$ is a family of $\lceil c\cdot (n+1)\ln n\rceil$ uniformly chosen pairs.
	\item For each vertex $v\in V$, and $d\in D$ we denote by $\mathcal{F}_2^{(v,d)}$ a sample of $\lceil c\cdot 4 e^2 \ln n\rceil$ sets sampled according to $\T{\frac{1}{d}}{V\setminus\set{v}}$ and define $\mathcal{F}_2$ as follows.  
	$$\mathcal{F}_2 = \bigcup_{v\in V,\ d\in D}\bigcup_{S\in \mathcal{F}_2^{(v,d)}} \singleedgep{S}{v}$$
	\item We denote by $\mathcal{F}_3^0$ a sample of $\lceil c\cdot 2 e^2 n\ln n\rceil$ sets sampled according to $\T{\frac{1}{\sqrt{n}}}{V}$, and define $\mathcal{F}_3$ as follows. $$\mathcal{F}_3 = \bigcup_{S\in \mathcal{F}_3^0}\singleedge{S}.$$
\end{itemize}
By Proposition~\ref{pro:single_edge}, we have that $|\singleedge{S}|\leq 24\ln |S|$ and the total number of queries is then $|\mathcal{F}_1\cup\mathcal{F}_2\cup\mathcal{F}_3| \leq c\cdot \alpha n\log^3 n$ for $\alpha=500$.

\begin{lemma} \label{thm:URgeneral}
	For every $c>0$, there is an algorithm which uses $c\cdot \alpha n\log^3 n$ queries and finds an edge with probability at least $1-n^{-c}$.
	\end{lemma}
\begin{proof}
	For a given graph $G = (V,E)$ with $n$ vertices, $m$ edges, and unknown maximal degree $d$, we sample $\mathcal{F}_1,\mathcal{F}_2,\mathcal{F}_3$ as described above and divide our analysis into three disjoint cases:
\begin{enumerate}
	\item $m>\frac{n}{2}$
	\item $m\leq \frac{n}{2}$ and $d\geq \sqrt{n}$
	\item $m\leq \frac{n}{2}$ and $d<\sqrt{n}$
\end{enumerate}
	In the first case, the probability for a query $S\in{\cal F}_1$ to hit an edge is ${m}/{\binom{n}{2}}> 1/(n+1)$, and when counting over all samples, the failing probability is at most $\left(1-\frac{1}{n+1}\right)^{|{\cal F}_1|} < n^{-c}$.
	
	In the second case, if ${\cal F}_1$ finds an edge, the algorithm succeed. Else, we use the answers from ${\cal F}_2$ in order to find a vertex of maximal degree $d$ and one of its neighbors.
	Given $v\in V$ of degree $d$ and $d'\in D$ such that $d'\leq d<2d'$, we bound the probability that $\mathcal{F}_2^{(v,d')}$ contains a set with a single edge. The probability of including exactly one of $v$'s neighbors is $$d\cdot \frac{1}{d'}\left(1-\frac{1}{d'}\right)^{d-1} > d\cdot \frac{1}{2d}\left(1-\frac{1}{d}\right)^{d-1}>\frac{1}{2e}.$$ Conditioning on a selection of a neighbor $u$, the probability of excluding all of $u$'s neighbors is at least, $$\left(1-\frac{1}{d'}\right)^{d-1}>\left(1-\frac{1}{d}\right)^{d-1}>\frac{1}{e}.$$ This is since $d$ is a bound on $u$'s degree and $v$ is a neighbor of $u$. In addition, for any other edge $e\in E$, not touching either $v$ or $u$, the probability for including $e$ in the sample is $\frac{1}{d'^2}\leq\frac{1}{n}$. Using the union bound, with probability of at least $1-\frac{m}{n}\geq\frac{1}{2}$, no such edge is selected.
	That is, for any set $S\in\mathcal{F}_2^{(v,d')}$ we have a probability of at least $\frac{1}{4e^2}$ that $S$ contains a single edge. Using the queries $\singleedge{S}$ the algorithm finds if this is indeed the case, and if so, returns an edge.
	Since there are $c\cdot 4 e^2 \ln n$ samples in $\mathcal{F}_2^{(v,d')}$, the failing probability is at most $n^{-c}$.
	%Since the family $\singleedge{A}$ returns a failing message we know if we are good...
	
	%In the first two families of queries, the algorithm either succeed in finding an edge or returned a failing message. In our last case, there is a probability of returning a wrong pair of vertices.
	In the last case, if ${\cal F}_1$ or ${\cal F}_2$ finds an edge, the algorithm succeed, and else we use the answers from ${\cal F}_3$ in the following way. For any pair $\set{u,v}\in E$, the probability for a set $S\in {\cal F}_3^0$ to include both $u$ and $v$ is $\frac{1}{n}$. 
	Conditioning on the selection of $u$ and $v$, the probability of excluding all of their neighbors from the sampling is at least
	$$\left(1-\frac{1}{\sqrt{n}}\right)^{2\sqrt{n}-2}>\frac{1}{e^2}.$$
	This is since the maximal degree of each of the vertices is smaller than $\sqrt{n}$.
	In addition, for any other edge $e\in E$ not touching either $v$ or $u$, the probability of including $e$ in the sample is $\frac{1}{n}$. Using a union bound, with probability of at least $1-\frac{m}{n}\geq\frac{1}{2}$, no such edge is selected.
	That is, for any $S\in {\cal F}_3^0$ there is a probability of at least $(2e^2\cdot n)^{-1}$ to contain exactly one edge. 
	Using the queries $\singleedge{S}$, the algorithm finds if this is indeed the case, and if so, to return an edge. Since there are $c\cdot 2 e^2 n\ln n$ samples in ${\cal F}_3^0$, the failing probability is at most $n^{-c}$.
	
	In total, the algorithm makes $|\mathcal{F}_1\cup\mathcal{F}_2\cup\mathcal{F}_3| \leq c\cdot \alpha n\log^3 n$ queries and fails to find an edge with probability at most $n^{-c}$. 
\end{proof}
Notice that failing is always of the form of an error message, i.e., the algorithm never return a pair which does not induces an edge.
\subsection{Lower Bound for General graphs}
We now complete the proof of Theorem~\ref{thm:General} with a matching lower bound.

\begin{restatable}{rLem}{LRgeneral}\label{thm:LRgeneral}
	Any randomized algorithm for edge finding that succeeds in a constant probability, requires $\Om{\frac{n}{\log^2 n}}$ queries. Moreover, the lower bound holds for the family of Star graphs.
\end{restatable}
\begin{proof}
	By Yao's principle \cite{Yao77Probabilistic}, it if sufficient to show that there exists a distribution $\D$ over the inputs such that no deterministic algorithm could succeed with high probability when the inputs are distributed according to $\D$. 
	%\ron{cite more formally? there is a nice explanation in Thm2 in \cite{kawachi2012query}}.
	Let $\D$ be the following distribution over graphs: we choose one vertex $v$ uniformly at random. For every $u\in V\setminus \set{v}$, the edge $\set{v,u}$ is in $E$ with probability $\frac{1}{\log n}$.
	
	Let $ALG$ be a deterministic algorithm for the problem who makes $o\left(\frac{n}{\log^2 n}\right)$ queries. We show that $ALG$ fails with high probability where inputs are drawn from $\D$.
	
	We have that any query of size at most $10\log^2 n$, includes $v$ with probability at most $\frac{10\log^2 n}{n}$. Since there are $o(\frac{n}{\log^2 n})$ queries in total, the probability that any query includes $v$ is negligible, and with probability $1-o(1)$ all of them return a negative answer.
	On the other hand, for every query $S$ of size at least $10\log^2 n$, the probability that $S$ does not include any vertex connected to $v$ is bounded by $(1-\frac{1}{\log n})^{10\log^2 n} < n^{-10}$. Since there are $o(\frac{n}{\log^2 n})$ queries, by using the union bound, we get that with probability of at least $1-n^{-9}$ over $\D$, in all of such queries, there is always a vertex $u\in S$ who is a neighbor of $v$. Thus, with probability $1-n^{-9}$, we have that $Q(S)$ is a mere indicator whether or not $v\in S$.
	%, and we cannot learn anything about the other vertices.
	In case that any of the queries acted unexpectedly, we assume the algorithm succeeds in finding an edge. Otherwise, the algorithm run on two different stars that are both rooted at $v$ will be identical regardless of the edges realization and the algorithm cannot do anything but guessing a random vertex and succeed with probability $\frac{1}{\log n}$.
	In total, we have a probability of $o(1)$ for finding an edge, hence proving a lower bound for the randomized query complexity. 
\end{proof}
The combination of the claims of Lemmas \ref{thm:URgeneral} and \ref{thm:LRgeneral} is exactly the statement of Theorem~\ref{thm:General}. 
	
\section{Finding an Edge in Special Cases}\label{sec:special}	
	In this section, we study families of graphs for which we design more efficient algorithms or prove tighter lower bounds, either in the deterministic or randomized models. All of the results in this Section are summarized in Theorem~\ref{thm:special}.

\subsection{Overlapping-Product graphs}
For this family, since star graphs are a type of overlapping-product graphs, the randomized lower bound of Theorem~\ref{thm:General} holds.
We complete this result by showing a deterministic algorithm with query complexity which is also quasi-linear.
The algorithm is based on a divide-and-conquer technique and it is described in the following proof.

\begin{lemma}\label{thm:UDstar}
	There exists a deterministic algorithm for finding an edge in a Overlapping-Product graphs, that uses $n\lceil \log n \rceil$ queries.
\end{lemma}
\begin{proof}
We build the algorithm recursively, assuming $n$ in the number of vertices in the graph.
For $n = 0,1$, the claim is trivial. 
For $n>1$, the algorithm divides the set $V$ into two disjoint sets $V_1, V_2$, such that $\left||V_1|-|V_2|\right|\leq 1$. Then, by the induction assumption, solve for each set recursively in parallel using $|V_1|\lceil \log |V_1| \rceil + |V_2|\lceil \log |V_2| \rceil \leq n\lceil \log n -1 \rceil$ queries.
If there is an edge that both of its endpoints in the same set, the algorithm finds it.

Otherwise, we check if there is an edge with one endpoint in each set, using the fact that both sets do not contain any edge and the "almost product" property.
We then query each of the sets in  
$$
\mathcal{F}_1 = \set{ \set{v}\cup V_2  \mid v\in V_1} \quad \mathcal{F}_2 = \set{ \set{v}\cup V_1  \mid v\in V_2}.
$$
For any $v\in V_1$ such that $Q(\set{v}\cup V_2)$ is positive, we have that $v$ is an endpoint of an edge since there are no edges in $V_2$.
Assume without loss of generality that $v\in A$, thus there exist a vertex $u\in V_2\cap B$.
For any $w\in V_2$ such that $Q(\set{w}\cup V_1)$ is positive, we have that $w\in B$. Otherwise, the pair $u,w$ induces an edge.
In the same way, we have that all vertices in $V_1$ which are endpoints of edges, belong to $A$
That is, any pair $v_1\in V_1$ and $v_2\in V_2$ such that $Q(\set{v_1}\cup V_2)$ and $Q(\set{v_2}\cup V_1)$ are positive, induced an edge.

Since all queries are asked in a single parallel querying round, the total number of queries is then $n\lceil \log n \rceil$.
\end{proof}

\subsection{Stars}
Lemma~\ref{thm:LRgeneral} shows that any randomized algorithm for finding an edge in a star graph must query at least $\Om{\frac{n}{\log^2 n}}$ queries. We now show a tighter lower bound for the deterministic case, using a reduction from the problem of witness finding over $n-1$ items.
\begin{lemma}
	Any deterministic algorithm for edge finding in a hidden Star graph, requires $n-1$ queries.
\end{lemma}
\begin{proof}
	Given an instance for the witness finding problem with $n-1$ items (or equivalently, an input $G_0=(V_0,E_0)$ for the edge finding problem where $\binom{V_0}{2} = n-1$), we show a reduction to the problem of edge finding in a star graph $G=(V,E)$ with $n$ vertices.
	
	denote by $v_1,\ldots,v_n$ the vertices of $V$ and by $u_1,\ldots,u_{n-1}$ the items (or edges) in the original problem.
	We define $E$ to have $v_n$ as the center of the graph and the neighbors are according to the original problem, i.e., $\set{v_i,v_n}\in\E$ if $u_i$ is positive in the original problem, and there are no other edges in $E$.
	
	Every query that does not include $v_n$, does not include an edge and thus always negative and gives no information.
	Every query that includes $v_n$ is positive if and only if $u_i$ is positive in the original problem, thus completing the reduction.	
	Lemma~\ref{thm:LDgeneral} then complete the claim.
\end{proof}

\subsection{Cliques}
The algorithm of Lemma~\ref{thm:UDstar} finds an edge in a clique using $\O{n\log n}$ queries, and can be improved to a linear query complexity, as we show on Appendix~\ref{sec:cliqueapp}.
\begin{restatable}{rLem}{UDclique}\label{thm:UDclique}
	There exist a deterministic algorithm for finding an edge in a hidden clique, that uses $\O{n}$ queries.
\end{restatable}
On the other hand, we now show that in the deterministic model, at least $\frac{n}{2}$ queries are needed in order to find an edge in a clique.

\begin{lemma}
	in the deterministic model, at least $\frac{n}{2}$ queries are needed in order to find an edge in a clique.
\end{lemma}
\begin{proof}
	Given a set of at most $\frac{n}{2}-1$ queries, $\mathcal{F}$, an adversary may build a fooling set in the following way:
	We keep a set $C$ which we initiate to include all the vertices of the graph. 
	While there still exists a set $S\in \mathcal{F}$ of size $1$ or $2$, update $S' = S'\setminus S$ for all $S'\in \mathcal{F}$, and $C = C\setminus S$.
	Since there are at most $\frac{n}{2}-1$ queries, and each set correspond to at most two vertices removal, the process ends with some set $C$ of size at least $2$, and all $S\in\mathcal{F}$ are either empty, or contains at least three vertices.
	According to the graph induced by $C$, the answers to the queries are negative on any set $S$ that is now empty, and positive for all other sets.
	No deterministic algorithm can distinguish between the graphs induced by $C$ and a graphed induced by $C\setminus\set{v}$ for any edge $v\in C$, thus failing on at least one input.
\end{proof}

On the other hand, we show that using randomized queries, an edge of the clique can be found using polylogarithmic many queries and polynomially small failing probability.
%\ronedit{The technique is similar to the one presented by \cite{ben1992theory} for the witness finding problem.}
Let $D = \set{2^i \mid i=1,...,\left\lceil\log n\right\rceil }$ be a set of estimations for the number of vertices in the hidden clique and $c>0$ some constant.
For each estimation $d\in D$, we sample $\lceil c\cdot 8e \ln n \rceil$ sets according to $\T{\frac{1}{d}}{V}$ and denote this family of samples $\mathcal{F}_d$. We then query
$$
\mathcal{F} = \bigcup_{d\in D}\bigcup_{S\in \mathcal{F}_d} \singleedge{S},
$$
where $|\mathcal{F}|\leq c \alpha \log^3 n$ for $\alpha = 500$.
\begin{lemma}
	For every $c>0$, there is an algorithm which makes $c \alpha \log^3 n$ queries and finds an edge in a Clique graph with probability at least $1-n^{-c}$.
\end{lemma}
\begin{proof}
 	Let $d$ be the number of vertices in the clique, and $d'\in D$ is such that $d'\leq d<2d'$. 
 	The algorithm succeed whenever there exist at least one sample in $\mathcal{F}_{d'}$ that has exactly one edge.
 	The probability the exactly two vertices of the clique are selected is $\binom{d}{2}\left(\frac{1}{d'}\right)^2\left(1-\frac{1}{d'}\right)^{d-2}\geq\binom{d}{2}\left(\frac{1}{2d}\right)^2\left(1-\frac{1}{d}\right)^{d-2}>\frac{1}{8e}$.
 	Since $|\mathcal{F}_{d'}|\geq c\cdot 8e \ln n$ we have that the failing probability is at most $n^{-c}$.
\end{proof}

\subsection{Matchings}
For Matching graphs, the algorithm presented in Lemma~\ref{thm:UDstar} fails and we now show a modification of the algorithm for this case.
\begin{lemma}
	There exists a deterministic algorithm for finding an edge in a hidden matchings, that makes $\O{n\log^2 n}$ queries.
\end{lemma}
\begin{proof}
	We build the algorithm recursively, assuming $n$ in the number of vertices in the graph.
	For $n = 0,1$, the claim is trivial. 
	For $n>1$, the algorithm divides the set $V$ into two disjoint sets $V_1, V_2$, such that $\left||V_1|-|V_2|\right|\leq 1$. Then, by the induction assumption, solves for each set recursively.
	If there is an edge that both of its endpoints in the same set, the algorithm finds it.
	Otherwise, we check if there is an edge with one endpoint in each set, using the fact the both sets do not contain any edge.
	We query for each $v\in V_1$ if it has a match in $V_2$, and if so, what it is. That is,
	$$
	\mathcal{F} = \bigcup_{v\in V_1} \singleedgep{V_2}{v},
	$$
	Let $f(n)$ be the total number of queries in the algorithm. We have that
	$$f(n) = f\left(\left\lceil\frac{n}{2}\right\rceil\right)+f\left(\left\lfloor\frac{n}{2}\right\rfloor\right)+\left\lfloor\frac{n}{2}\right\rfloor\cdot 24\ln \left\lceil\frac{n}{2}\right\rceil $$
	for $n>1$ and $f(1)=1$ trivially.
	Using the generalized Master theorem (see, e.g., \cite{cormen2009introduction}) we have that $f(n) = \O{n\log^2 n}$, when all of the queries are asked in a single parallel querying round.
\end{proof}

We show a deterministic lower bound of $\frac{n}{2}$ queries using reduction from the witness finding problem.
The reduction also proves that any randomized algorithm that succeeds with constant probability, makes $\Om{\log^2 n}$ queries.
\begin{lemma}\label{lem:Lmatching}
	Any deterministic algorithm for edge finding in a hidden matching, requires $\frac{n}{2}$ queries. Moreover, Any randomized algorithm for this problem who succeeds in finding an edge with constant probability, requires $\Om{\log ^2 n}$ queries.
\end{lemma}
\begin{proof}
	Given an instance for the witness finding problem with $\frac{n}{2}$ items (or equivalently, an input $G_0=(V_0,E_0)$ for the edge finding problem where $\binom{V_0}{2} = \frac{n}{2}$), we show a reduction to the problem of edge finding in a Matching graph $G=(V,E)$ with $n$ vertices.
	
	denote by $v_1,\ldots,v_{\frac{n}{2}},u_1,\ldots,u_{\frac{n}{2}}$ the vertices of $V$ and by $w_1,\ldots,w_{\frac{n}{2}}$ the items (or edges) in the original problem.
	We define $E$ to have an edge $\set{v_i,u_i}$ in $E$ if $w_i$ is positive in the original problem, and there are no other edges in $E$.
	Given a query $S$, the answer is positive if and only if there exists $i$ such that both $v_i$ and $u_i$ belong to $S$ and $w_i$ is positive.
	That is, any query for this graph can be asked on $G_0=(V_0,E_0)$ and the reduction is complete.
	Lemma~\ref{thm:LDgeneral} then show our claim for the deterministic model.
	For the randomized model, the lower bound of $\Om{\log^2 n}$ for witness finding is due to Kawachi et. al. \cite{kawachi2012query}.
\end{proof}

In the randomized model, we take a similar approach to the one taken for Clique graphs.
Let $D = \set{2^i \mid i=1,...,\left\lfloor\log n\right\rfloor }$ be a set of estimations for the number of edges in the hidden matching and $c>0$ some constant.
For each estimation $d\in D$, we sample $ \lceil c\cdot 2e \ln n \rceil$ sets according to $\T{\frac{1}{\sqrt{d}}}{V}$ and denote this family of samples $\mathcal{F}_d$. We then query
$$
\mathcal{F} = \bigcup_{d\in D}\bigcup_{S\in \mathcal{F}_d} \singleedge{S},
$$
where $|\mathcal{F}|\leq c \alpha \log^3 n$ for $\alpha = 500$.

\begin{lemma}
 For every $c>0$, there is an algorithm which makes $c \alpha \log^2 n$ queries and finds an edge in a Matchings graph with probability at least $1-n^{-c}$.
\end{lemma}
\begin{proof}
	Let $d$ be the number of edges in $E$, and $d'\in D$ is such that $d'\leq d<2d'$. 
 	the algorithm succeeds if at least one of the samples in $\mathcal{F}_{d'}$ has exactly one edge.
 	$$d\left(\frac{1}{\sqrt{d'}}\right)^2\left(1-\left(\frac{1}{\sqrt{d'}}\right)^2\right)^{d-1}\geq d\cdot\frac{1}{2d}\left(1-\frac{1}{d}\right)^{d-1}>\frac{1}{2e}.$$
 	Since $|\mathcal{F}_{d'}|\geq c\cdot 2e \ln n$ we have that the failing probability is at most $n^{-c}$.
\end{proof}
Note that by selecting $c = 1/\ln n$, there is constant probability for the algorithm to succeed while making only $\O{\log^2 n}$ queries, which match the lower bound of Lemma~\ref{lem:Lmatching}.
This conclude the proof of Theorem~\ref{thm:special}.
		
\section{Edge Finding Using $r$ Rounds}\label{sec:rounds}	
	In this section we consider an adaptive model of $r$ rounds, when at each round a family of queries are asked and answered in parallel and the selection of the queries depends only on answers from previous rounds and possibly randomization.
When limited to $r$ adaptive querying rounds, we use the non-adaptive algorithms from Section~\ref{sec:general} in order to construct algorithms with improved query complexity.
We start by stating a lower bound by \cite{gerbner2016rounds} for the deterministic case. This bound holds even for arbitrary queries on any subset of edges. 
\begin{restatable}[\cite{gerbner2016rounds}]{rLem}{LDrounds}\label{thm:LDrounds}
	Any deterministic algorithm for edge finding which uses at most $r$ rounds, requires $\Om{r\cdot n^{2/r}}$ queries.
\end{restatable}

\begin{definition}
	For a graph $G = (V,E)$ and a disjoint partition of the vertices $\set{V_i}$, we define the {\em partition graph}, $H = (\mathcal{V},\mathcal{E})$ where $\mathcal{V} = \set{V_i}$ and $$\mathcal{E} = \set{ \set{U,V} \mid U,V\in \mathcal{V}, \exists a,b\in U\cup V,\ s.t.\ \set{a,b}\in E }.$$
\end{definition}

\rounds*
\begin{proof}	
	In the deterministic model, assume by induction that for any $r'<r$ we have that the total number of queries is at most $10 r'\cdot n^{2/r'}$. The case $r=1$ is handled in Theorem~\ref{thm:General}.
	For $r>1$, partition $V$ arbitrarily into $k=\left\lceil 2\cdot n^{1/r} \right\rceil+4$ sets $V_1,\ldots,V_{k}$ of size at most $t = \left\lfloor \frac{1}{2}\cdot n^{1-1/r} \right\rfloor$ each \footnote{We have that $kt\geq n$ for $r>1$.}.
	In the partition graph, we query all of the queries $Q(\set{V_i,V_j})$ (i.e., $Q(V_i\cup V_j)$ in the original graph), for $i,j\in[k]$. If all of the answers are negative, the algorithm terminates, otherwise it continues with any pair of two positive sets, iteratively.
	By the induction assumption, the total number of queries $f(n)$ is then,
	\begin{eqnarray*}
	 f(n) & \leq &\binom{k}{2} + 10(r-1){(2t)}^{2/(r-1)}\\
	 	& \leq &\left(2\cdot n^{1/r}+5\right)\left(2\cdot n^{1/r}+4\right)/2 + 3(r-1){\left(n^{1-1/r}\right)}^{2/(r-1)}\\
		& \leq & 2\cdot n^{2/r} +  9\cdot n^{1/r} + 10 + 10(r-1)\cdot n^{2/r}\\
		& \leq & 10r\cdot n^{2/r},
	\end{eqnarray*}
	where the last inequality holds whenever $r<\log n$. For $r\geq \log n$, the problem is solvable using binary search algorithm that makes $6$ queries at each round (see Appendix~\ref{sec:binary}).
	
	In the randomized model, assume that for any $r'<r$ we have that the total number of queries is at most $2000 c r'\cdot n^{1/r'}\ln^3 n$ and the error probability is at most $r\cdot n^{-c/r}$. The case $r=1$ is handled Theorem~\ref{thm:General}.
	For $r>1$, partition $V$ arbitrarily into $k=\left\lceil 2\cdot n^{1/r} \right\rceil+4$ sets $V_1,\ldots,V_{k}$ of size at most $t = \left\lfloor \frac{1}{2}\cdot n^{1-1/r} \right\rfloor$ each.
	In the partition graph, we run the single-round randomized algorithm from Theorem~\ref{thm:General} with error parameter of $c$.
	The running on the partition set is positive whenever there is an edge in the union of two sets. We continue with this union iteratively, using $r-1$ rounds.
	By the induction assumption the total number of queries $f(n)$ is then,
	\begin{eqnarray*}
		f(n) & \leq & 500 c k \ln^3 k + 2000c(r-1){(2t)}^{1/(r-1)}\ln^3 2t\\
			& \leq &  500 c (2n^{1/r} + 5) \ln^3 n + 2000c(r-1){n}^{1/r}\ln^3 n\\
			&\leq & 2000cr{n}^{1/r}\ln^3 n
	\end{eqnarray*}
	where the second inequality holds whenever $r<\log n$. For $r\geq \log n$, the problem is solvable with standard binary search algorithm, using two queries at each round.
	
	The algorithm errs with probability at most $k^{-c} + (r-1)\cdot (2t)^{-c} \leq r\cdot n^{-c/r}$.
	Taking any $c>2r$ gives a polynomially small error probability, with the required query complexity for $r<\log n$. As mentioned, for $r\geq \log n$ the claim hold deterministically.
		
\end{proof}

\section{Discussion and Open Problems}\label{sec:discussion}	
	In this work we showed a randomized algorithm for edge finding in the non-adaptive IS model and prove the tightness of its query complexity. We showed that while in the deterministic model the optimal algorithm requires $\binom{n}{2}$ queries, in the randomized model $\Tt{n}$ queries are sufficient (and needed) in order to find an edge.
In addition, we analyzed the query complexity for Stars, Clique and Matching graphs, for both the randomized and deterministic model.
Lastly, for general graphs, we showed a trade-off between the query complexity and the number of adaptive rounds $r$ made by the algorithm. We show two algorithms with $\O{rn^{2/r}}$ and $\Ot{rn^{1/r}}$ sample complexity for the deterministic and randomized models, respectively.

Two future directions are:
\begin{itemize}
	\item Closing the polylogarithmic gaps between lower and upper bounds in some of our results.
	\item Extending our results for edge finding in a graph to the task of hyperedge finding in a hypergraphs, with the appropriate query model as discussed in \cite{angluin2006learning,abasi2018non}. This extension has implications for welfare maximization as discussed in Appendix~\ref{app:oxs}.
	\item Finding the optimal adaptive randomized bound is still an open question. It would be interesting to see either a $\Omega{rn^{1/r}}$ lower bound or an improved randomized algorithm that takes further advantage of adaptivity.
\end{itemize}

%%
%% Bibliography
%%

%% Please use bibtex, 

	\bibliography{main} 
	\bibliographystyle{plainnat}			

\appendix
	\section{Maximization under cardinality constraint}\label{app:oxs}
We study the sample complexity maximizing the value of a valuation $f:2^{[n]}\rightarrow \R$ subject to a cardinality constraint $k$, i.e., to find $\arg\max_{|S|=k} f(F)$.
For simplicity, we assume (with loss of generality) that the $f$ is capped in $k$ items, i.e, the value for a bundle $S$ of size larger than $k$ is exactly $$f(S) = \max_{\begin{matrix} T\subset S,~ |T|=k \end{matrix}} f(T).$$
Observe that under this assumption, the value of the grand bundle is the same as the value of an optimal bundle of size $k$, and every value query for a bundle $S$ hinged whether or not the set contains an optimal solution.
We identify every optimal solution with an hyper-edge of dimension $k$ in an hyper-graph with $n$ vertices.
For $k=2$, we receive a graph $G$ where any edge in the graph is a feasible solution of the maximization problem, and the connection to the edge-finding-problem discussed in this paper is clear.

We now focus on the following family of valuations called OXS, and ask for the number of non-adaptive value queries needed in order of solving the maximization problem.
\begin{definition}
	A function $f:2^{[n]}\rightarrow \R$ is an {\it assignment function (OXS)} if $f$ is the convolution of $r$ unit-demand functions $u_1, \ldots, u_r$:
	$f(S)  = \bigvee_{i \in [r]} u_i(S) := \max_{\cupdot_{i\in [r]}{S_i} =S}\sum_{i \in [r]}{u_i(S_i)}$,
	where the sets $S_1, \ldots S_r$ are a partition of $S$.
\end{definition}

When limiting ourself to OXS valuations characterized by two unit-demand valuations, i.e., $r=2$, we observe that the capping at $k=2$ property trivially holds, and that the graph induced by an OXS valuations has a special form: it belongs to the family of Overlapping-Product graphs defined in Section~\ref{sec:model}.
\begin{proposition}
	Given an OXS valuation $f$ described by two unit-demand valuations, the graph induced by pairs of items of maximal value as a bundle is an Overlapping-Product graph.
\end{proposition}
\begin{proof}
	By definition there exist two unit-demand valuations $u_1,u_2$ such that $f(S)  = \bigvee_{i \in [2]} u_i(S)$.
	Let $A_i^j$ be the set of items of $j$th highest values in $u_i$, and denote the value of item in $A_i^j$ by $a_i^j$. 
	
	If there exist a pair $x\neq y$ such that $x\in A_1^1$ and $y\in A_2^1$, then in any maximal pair both valuations are maximized and a pair $r,s$ is maximizing $f$ if and only if one of them belongs to $A_1^1$ and the other to $A_2^1$ (might be the case that some of them belong to both sets). This is an Overlapping-Product graph with those sets.
	
	In no such pair exists, we have that $A_1^1 = A_2^1$ and this set is a singleton. In this case, any maximizing pair includes this item with addition of an item from some set ( either $A_1^2$, $A_2^2$ or their union, depending on the values of $a_1^2$ and $a_2^2$.). This is also an Overlapping-Product graph (in fact, this is a star graph).
\end{proof}
An immediate corollary is that the non-adaptive complexity of maximizing an OXS valuation with cardinality constraint $k=2$, is at most $\O{n\log n}$.	
	\section{Missing Proofs}\label{app:missing}
\subsection{Explicit constructions for Proposition~\ref{pro:single_edge}}\label{sec:exlicit}
\begin{lemma}
	There exist a deterministic IS algorithm that for a Singleton graph, returns the only $e\in E$, using $\lceil4\log n\rceil$ non-adaptive deterministic queries.
\end{lemma}
\begin{proof}
	We name the vertices of $V$ using $k = \lceil\log n\rceil$ bits, define $$F_i^b = \set{v\in V\mid\text{the $i$-th bit of $v$ is $b$}},$$ and query
	$$\mathcal{F} = \set{ F_i^b \mid i\in\left[\lceil\log n\rceil\right],\ b\in\set{0,1}} .$$
	Assuming $G$ has a single edge $\set{u,v}$, each positive answer $Q(F_i^b)$ ensures that $u$ and $v$ have the same value for this bit, and its different otherwise. 
	If we look at the vertices of $V$ as elements of the filed $\mathbb{F}_{2^k}$, this set of answers encode an equation of the form $u\oplus v = a$ for some $a\in \mathbb{F}_{2^k}$.
	
	We now map each vertex $v\mapsto v^{-1}$ and query according to the new names.
	That is, we have that
	$$H_i^b = \set{v\in V\mid\text{the $i$-th bit of $v^{-1}$ is $b$}},$$ and we query
	$$\mathcal{H} = \set{ H_i^b \mid i\in\left[\lceil\log n\rceil\right],\ b\in\set{0,1}} .$$
	Hence, we can now deduce an equation of the form $u^{-1}\oplus v^{-1}= b$ for some $b\in \mathbb{F}_{2^k}$.
	
	The set of equations give the values of $u\oplus v$ and $u\cdot v$ and have a unique solution up to switching $u$'s and $v$'s names.
	
	By additionally querying $Q(V)$, we can verify that $E$ is not empty.
\end{proof}

For cases in which it is known that $G$ is star graph with center at a known vertex $v$, we describe an algorithm that uses $\lceil2\log n\rceil$ queries for the same problem.
\begin{lemma}
	There exist a deterministic IS algorithm that for a Star graph $G = (V,E)$ with $m$ edges and a known center $v$, returns the only $e\in E$, if $m=1$,
	in case that $m=0$ the algorithm returns the message "$none$", in case that $m>1$ returns the message "$more\ than\ one$", and it is implementable using $\lceil2\log n\rceil$ non-adaptive deterministic queries.
\end{lemma}
\begin{proof}
	We name the vertices of $U$ using $\lceil\log n\rceil$ bits and query
	$$\mathcal{F} = \set{ \set{v}\cup\set{u\in U\mid\text{the $i$-th bit of $u$ is $b$}}  \mid i\in[\lceil\log n\rceil],\ b\in\set{0,1}} .$$
	The positive answers encode the endpoint in $U$. If there is no such endpoint, all answers would be negative. While if there are more then one such endpoint, there is an index $i$ with positive answer for both values of $b$. Either way, the algorithm identify this is the case and return the proper message.
\end{proof}

\subsection{Proof of Lemma~\ref{thm:LDgeneral}}\label{sec:LDG}
\LDgeneral*
\begin{proof}
	Given a set of queries, an adversary may build a fooling set in the following way:
	We start with $G = (V,E)$, with $E=\binom{V}{2}$, the full graph over $V$ vertices.
	For every query $S$, let $E(S)$ be the set of edges with both endpoints in $S$.
	While there still exist a set $S$ such that $|E(S)|$ of size $1$, update $E(S') \gets E(S')\setminus E(S)$ for all $S'$, and $E \gets E\setminus E(S)$. That is, we guarantee that the answer for the query $S$ is negative.
	
	Since there are at most $\binom{n}{2}$ queries, and each set correspond to a single edge removal, the process end with some set $E$ and for any $S$ we have that $E(S)$ are either empty or contains at least two edges. In addition, $E$ is not empty.
	
	According to the graph induced by $E$, the answers to the queries are negative on any set $S$ such that $E(S)$ is empty, and positive for all other sets.
	Hence, no deterministic algorithm can distinguish between the graphs induced by $E$ and a graphed induced by $E\setminus\set{e}$ for any edge $e\in E$, thus failing on at least one input.
\end{proof}

\subsection{Tighter Upper Bounds for Cliques}\label{sec:cliqueapp}
\UDclique*
\begin{proof}
	We describe three different algorithms, all using $\O{n}$ queries.
	
	First, as in Lemma~\ref{thm:UDstar}, we use a divide-and-conquer technique.
	We divide the set $V$ into two disjoint sets $V_1, V_2$, such that $\left||V_1|-|V_2|\right|\leq 1$, and solve for each set recursively.
	Since $G$ is a clique graph, if both $V_1$ and $V_2$ induce no edge, there are at most one vertex from $S$ in either one. That is, the graph contains at most one edge and it can be found by querying $\singleedge{V}$ according to the construction in Appendix~\ref{sec:exlicit}.
	the total number of queries $f(n)$ is then following the formula:
	$$f(n) = f\left(\left\lceil\frac{n}{2}\right\rceil\right)+f\left(\left\lfloor\frac{n}{2}\right\rfloor\right)+\lceil4\log n\rceil+1.$$
	Using Akra–Bazzi theorem \cite{akra1998solution}, we have that $f(n) = \Th{n}$.
	
	For the other two algorithms, we start by naming the vertices arbitrarily $1,\ldots,n$ and querying,
	$$\mathcal{F}_1 = \set{ \set{1,\ldots,i}  \mid i\in[n]} \quad \mathcal{F}_2 = \set{ \set{i,\ldots,n}  \mid i\in[n]}.$$
	Since $V$ is in both sets, the total number of queries is exactly $2n-1$. 
	
	We can now identify the second smallest and second largest indexes of vertices of the Clique, those are the vertices that correspond to the first positive answer in $\mathcal{F}_1$ and the last positive answer in $\mathcal{F}_2$, respectively. If they are not identical, the algorithm found an edge.
	The only case in which those vertices are the same, is when the clique size is exactly three. In this case, only one of the clique's vertices is known.
	
	We find a second vertex using one of the following two schemes.
	We may query:
	$$\mathcal{F}_3 = \set{ \set{{i},\ldots,{(i+\lfloor\frac{n}{2}\rfloor\mod n)}}  \mid i\in[n]} .$$
	Let $S_\ell = \set{1,...,\lfloor\frac{n}{2}\rfloor}$ and $S_h = \set{\lceil\frac{n}{2}\rceil+1,...,n}$.
	If $n$ is even, one of the sets contains at least two vertices of the clique while the other contains at most one. That is, for some of the sets in $\mathcal{F}_3$ the answer is positive and for some are not, and we can identify two of the clique's vertices by the phase shift in answers.
	If $n$ is odd, we still have that either $S_\ell$ or $S_h$ contains at most one of the clique vertices.
	If one of them contains at least two vertices, we continue as before, otherwise we have that $\lceil\frac{n}{2}\rceil$ is part of the clique. In the later case, the only way that $\lceil\frac{n}{2}\rceil$ is not included in a positive query is if the other two vertices are $1$ and $n$. But then, the query $\set{n, 1,\ldots,\lfloor\frac{n}{2}\rfloor-1}$ is positive, and again we can identify  two of the clique's vertices by the phase shift in answers.
	This algorithm has a total query complexity of $3n$.

	Another option is to use a similar technique as in Proposition~\ref{pro:single_edge} for finding a single edge, and use it to find all three edges of the graph.
	As before, we ask $2n$ queries in order to identify one of the clique vertices.
	In addition, we build a randomized set of queries by sampling $t=64\ln n$ sets of vertices according to $\T{\frac{1}{2}}{V}$.
	Assuming the vertices of the clique are $\set{a,b,c}\in V$, for every sample $S$, it holds that $Q(S)$ is positive if and only if at least two of the vertices are in $S$.
	Assume $a$ is known from the first part of the algorithm.
	For every vertices $a,b,c,d\in V$, the probability that $a,d\in S$ and $b,c\notin S$ is $\frac{1}{16}$, and the probability that this would not hold any samples is  
	$$\left( 1-\frac{1}{16}\right)^t < n^{-4}.$$ 
	Whenever we have a negative answer that includes $a$, we are guaranteed that any $d\in S$ is not a part of the clique. Taking a union bound over all possible $a,b,c,d\in V$, gives that with positive probability, using $t$ samples, for any Clique graph $\set{a,b,c}$, we can rule out all other vertices as being part of an edge.
	Hence, there exist some selection of the $t$ queries that is always correct, and we can find the edge deterministically using $t$ queries and the total number of queries is $2n+64\ln n$.
	
\end{proof}

\subsection{Fully Adaptive Binary Search}\label{sec:binary}
\begin{lemma}
	There exist a deterministic $\lceil\log n\rceil$-adaptive algorithm for finding an edge using $6$ queries at each round.
\end{lemma}
\begin{proof}
	By adding dummy vertices, we may assume that $n$ is power of $2$. We prove our claim by induction. For $n \leq 4$, the algorithm simply query all possible pairs. 
	For general $n>4$, we split the vertices arbitrarily into $4$ equal size sets $V_1,...,V_4$, and query all $6$ queries of the form $Q(V_i\cup V_j)$ for $\set{i,j}\in\binom{[4]}{2}$.
	Any edge in the graph is contained in at least one such query, and we can continue with any positive answer while removing half of the vertices. Thus, having $\log n - 1$ total rounds.
\end{proof}

\end{document}